\documentclass[twoside,leqno,onecolumn]{article}
\usepackage{amsmath,amsfonts,amssymb,xspace}
\usepackage{ltexpprt}
\usepackage{graphicx}
\usepackage{xcolor}
\usepackage{setspace,caption}
\usepackage{algorithm}
\usepackage{algorithmic}
\newcommand\mbR{\mbox{$\mathbb{R}$}}
\newcommand\mbC{\mbox{$\mathbb{C}$}}

\newcommand\ket[1]{| #1 \rangle}
\newcommand\bra[1]{\langle #1 |}
\newcommand\qip[2]{\langle #1 | #2 \rangle}

\newcommand\rank{\mbox{\tt {rank}}\xspace}
\newcommand\srank{\mbox{\tt {S-rank}}\xspace}

\newcommand\prank{\mbox{\tt {rank}$_{\tt psd}$}\xspace}

\newcommand\tr{\mbox{\tt {tr}}\xspace}

\newcommand\size{\mbox{\sf {size}}\xspace}
\newcommand\Q{\mbox{\sf {Q}}\xspace}
\newcommand\R{\mbox{\sf {R}}\xspace}

\newcommand\h{\mathcal{H}}

\mathchardef\mhyphen="2D

\newcounter{thm}
\newtheorem{example}[thm]{Example}

\newcommand{\comment}[1]{{}}

\allowdisplaybreaks[1]
\begin{document}

    \title{\vspace{-1cm} The Generations of Classical Correlations via Quantum Schemes}

    \author{Zhenyu Chen\thanks{Institute for Interdisciplinary Information Sciences, Tsinghua University, Beijing 100084, P. R. China. Email: chenzhen20@mails.tsinghua.edu.cn},~
    Lijinzhi Lin\thanks{Department of Computer Science and Technology, Tsinghua University, Beijing 100084, P. R. China. Email: lljz20@mails.tsinghua.edu.cn},~
    Xiaodie Lin\thanks{Institute for Interdisciplinary Information Sciences, Tsinghua University, Beijing 100084, P. R. China. Email: linxd19@mails.tsinghua.edu.cn},~
    Zhaohui Wei\thanks{Yau Mathematical Sciences Center, Tsinghua University, Beijing 100084, and Yanqi Lake Beijing Institute of Mathematical Sciences and Applications, 101407, P. R. China. Email: weizhaohui@gmail.com},~
    Penghui Yao\thanks{State Key Laboratory for Novel Software Technology, Nanjing University, Nanjing, Jiangsu Province 210023, P. R. China, and Hefei National Laboratory, Hefei 230088, P. R. China. Email: pyao@nju.edu.cn}}

    \maketitle

    \begin{abstract}
    Suppose two separated parties, Alice and Bob, share a bipartite quantum state or a classical correlation called a \emph{seed}, and they try to generate a target classical correlation by performing local quantum or classical operations on the seed, i.e., any communications are not allowed. We consider the following fundamental problem about this setting: whether Alice and Bob can use a given seed to generate a target classical correlation. We show that this problem has rich mathematical structures. Firstly, we prove that even if the seed is a pure bipartite state, the above decision problem is already NP-hard and a similar conclusion can also be drawn when the seed is also a classical correlation, implying that this problem is hard to solve generally. Furthermore, we prove that when the seed is a pure quantum state, solving the problem is equivalent to finding out whether the target classical correlation has some diagonal form of positive semi-definite factorizations that matches the seed pure state, revealing an interesting connection between the current problem and optimization theory. Based on this observation and other insights, we give several necessary conditions where the seed pure state has to satisfy to generate the target classical correlation, and it turns out that these conditions can also be generalized to the case that the seed is a mixed quantum state. Lastly, since diagonal forms of positive semi-definite factorizations play a crucial role in solving the problem, we develop an algorithm that can compute them for an arbitrary classical correlation, which has decent performance on the cases we test.
    \end{abstract}

    \section{Introduction}
    Local state transformation is a fundamental problem in quantum information theory, which has wide applications in communication complexity, resource theory, quantum distributed computing, quantum interactive proof systems, etc. In the setting of local state transformation, given two bipartite states  $\rho$ and  $\sigma$, two separated parties, Alice and Bob, share $\rho$, and their goal is to generate $\sigma$. In this paper, we concern ourselves about the case {where} the target state $\sigma$ is a {classical correlation} $(X,Y)$, which is a pair of random {variables distributed} according to a joint probability distribution $P$, for convenience, we call {half} of the total number of bits needed to record the labels of $(X,Y)$ the \emph{size} of $P$. Suppose $P$ is not a product probability distribution, then Alice and Bob have to share some {resource} beforehand, such as another classical correlation $P'$ called a \emph{seed} correlation or a shared quantum state $\rho$ called a \emph{seed} state. In the former case, all their operations are classical and local (each operation involves only one party), while in the latter case, they could perform local quantum operations, i.e., each party can make proper quantum measurements on the quantum subsystem and output the outcomes as $X$ or $Y$. Similar to the classical case, we define the size of $\rho$ to be half of the total number of qubits in $\rho$, denoted $\size(\rho)$. Here we stress that when Alice and Bob are generating a target classical correlation $P$ based on {their} shared resources, any communication between them is not allowed. In the current paper, we will focus on this kind of setting.

    In such settings, the minimum size of such seed correlation (state) that can generate $P$ has been defined as the \emph{classical (quantum) correlation complexity} of $P$, denoted $\R(P)$ ($\Q(P)$). In fact, it has been proved that \cite{zhang2012quantum,jain2013efficient}
    \begin{equation}
    \R(P) = \lceil \log_2 \rank_+(P) \rceil,
    \label{eq:qcorr+rank}
    \end{equation}
    and
    \begin{equation}
    \Q(P) = \lceil \log_2 \prank(P) \rceil.
    \label{eq:qcorrprank}
    \end{equation}
    For any nonnegative matrix $P\in \mbR_+^{n\times m}$, its nonnegative rank $\rank_+(P)$ is defined as the minimum number $r$ such that $P$ can be decomposed as the summation of $r$ nonnegative matrices of rank $1$, and $\prank(P)$ is its positive semi-definite rank (PSD-rank), which is the minimum $r$ such that there are $r \times r$ positive semi-definite matrices $C_x$, $D_y\in\mbC^{r\times r}$ satisfying that $P(x,y) = \tr (C_x  D_y)$ for all $x$ and $y$, and this is called a \emph{PSD decomposition} \cite{fiorini2012linear,fawzi2015positive}.

    In such a situation, a natural further question {arises: given a specific quantum state} $\rho_0$ and $\size(\rho_0)\geq\Q(P)$, can we generate $P$ by measuring $\rho_0$ locally? Similarly, if we have a classical correlation $P_0$ such that $\size(P_0)\geq\R(P)$, can we generate $P$ based on $P_0$ using local operations only? Note that the work in \cite{jain2013efficient} cannot answer these questions. In the current manuscript, we will show that these problems have very rich mathematical structures.

    \section*{Contributions}
    Our main contributions are as follows.
    \begin{itemize}
        \item We show that even for the special cases that the seed is a bipartite pure state $\ket{\psi}$ or another classical correlation $P'$, determining whether or not $P$ can be produced from $\ket{\psi}$ or $P'$ is NP-hard, implying that one-shot local state transformation is generally hard to solve.

        \item We prove that if the seed is a bipartite pure state $\ket{\psi}$, determining whether or not $P$ can be produced from $\ket{\psi}$ is equivalent to finding out whether $P$ as a nonnegative matrix has a certain {diagonal} form of {PSD factorizations} that matches $\ket{\psi}$, revealing an interesting connection between the current problem and optimization theory.

        \item We prove several necessary conditions that $\ket{\psi}$ has to satisfy to generate $P$. Furthermore, these conditions can also be generalized to the case that the seed is a mixed quantum state.

        \item We develop an algorithm that can compute {the {diagonal} form of PSD factorizations} for an arbitrary $P$, which has decent performance on the cases we test, and thus can be directly utilized to determine whether a given seed state $\ket{\psi}$ can produce $P$.
    \end{itemize}

   % Second,  Third, based on the above connection and other insights, we prove several necessary conditions that $\ket{\psi}$ has to satisfy to generate $P$. Furthermore, it turns out that these conditions can also be generalized to tackle the case that the seed is a mixed quantum state. Lastly, since \red{diagonal} forms of PSD factorization play a crucial role in the current problem,

    \section*{Related work}

    {\em Non-interactive simulations of joint distributions}, the classical counterpart of local state transformation, is a fundamental task in information theory. In this setting, two separated parties are provided with sequences $X^n$ and $Y^n$, respectively, where $\{(X_i,Y_i)\}_{1\leq i\leq n}$ are identically distributed draws from a probability distribution $P(x,y)$. The goal is to determine whether they can generate a pair of random variables $(U,V)_{u,v}$ possessing a joint distribution that closely approximates a target distribution $Q(u,v)$ without any communication between them. The history of research on non-interactive simulations of joint probability distributions dates back to the pioneering works by G\'acs and K\"orner~\cite{Gacs:1973} and Wyner~\cite{Wyner:1975:CIT:2263311.2268812}. Interested readers may refer to subsequent work such as \cite{7782969,7452414} and the references therein for more information. Local state transformations, first investigated by Beigi~\cite{Beigi:2013}, have gained increasing attention recently. However, characterizing the possibility of local state transformations has proved challenging, leading to various studies on necessary and sufficient conditions, as well as algorithm designs for this problem~{\cite{mojahedian2019correlation,qin2021nonlocal,chau2012entanglement,QYao:2022}}.
    % Most research on non-interactive simulations of joint distributions and local state transformations focuses on the asymptotic setting, where the parties share infinite copies of the source states {\cite{Delgosha2014,Beigi:2013}}. However, to the best of our knowledge, research on one-shot non-interactive simulations of joint probability distributions and one-shot local state transformations is much less. Jain, Shi, Wei, and Zhang proved that the minimum sizes of seed classical states or quantum states required to generate a given joint distribution are tightly captured by nonnegative ranks and positive semidefinite ranks respectively. However, both of them are NP-hard to compute~\cite{vavasis2010complexity,shitov2017complexity}.

    Most research on non-interactive simulations of joint distributions and local state transformations focuses on the asymptotic setting, where the parties share infinite copies of the source states {\cite{Delgosha2014,Beigi:2013}}. {In this scenario,  two powerful tools are studied: quantum maximal correlation \cite{Beigi:2013} and hypercontractivity ribbon \cite{Delgosha2014}. These are utilized to demonstrate nontrivial applications, particularly in proving the impossibility of local state transformations. However, one limitation of migrating these two approaches to the one-shot case is that these two quantities remain the same for all entangled pure states. For example, the quantum maximal correlation equals $1$ for all entangled pure states, which implies that they contain the strongest correlation, thus failing in providing nontrivial results in the one-shot case.}

    % In this work we focus one the setting of limited resources, say, only one copy of seed state is available. We find that entangled pure states are not equivalent and provide some nontrivial conditions which can not be given by \cite{Delgosha2014,Beigi:2013}.

    {To} the best of our knowledge, research on one-shot non-interactive simulations of joint probability distributions and one-shot local state transformations is much less. Jain, Shi, Wei, and Zhang proved that the minimum sizes of seed classical states or quantum states required to generate a given joint distribution are tightly captured by nonnegative ranks and positive semidefinite ranks respectively. However, both of them are NP-hard to compute~\cite{vavasis2010complexity,shitov2017complexity}.

    \section*{Preliminaries}
    Let $[n]=\{1,2, \ldots, n\}$. Let $A=[A(i, j)]_{i,j}$ be an arbitrary $m$-by-$n$ matrix with the $(i, j)$-th entry being $A(i, j)$, and we write $A^T$ as the transpose of A. We define $\text{diag}(x_1, x_2, \ldots, x_n)$ as the $n$-by-$n$ diagonal matrix with the diagonal entries $x_1, x_2, \ldots, x_n$. If A is a {self-adjoint} matrix, let $A = \sum_a a\ket{a}\bra{a} $ be its spectral decomposition. Given a function $f$ from complex numbers to complex numbers, we define $f(A) = \sum_af(a) \ket{a}\bra{a}$. A matrix $A$ is said to be {positive semi-definite} if all its eigenvalues are nonnegative, and we write $A\geq 0$ to indicate matrix $A$ is a PSD matrix. If $A,B \geq 0$ such that $\mathrm{tr}(AB)=0$, then $AB=0$ \cite{fawzi2015positive}.

    If $P$ is an $m$-by-$n$ classical correlation, let $P(x)$ denote the marginal probability of getting $x$, i.e., $P(x) \equiv \sum_y P(x,y)$, and similarly, $P(y) \equiv \sum_x P(x,y)$. $P(x|y) \equiv \frac{P(x,y)}{{P(y)}}$ denotes the conditional probability of getting $x$ given $y$. Define the classical fidelity for two nonnegative vectors $p(x)$ and $q(x)$ as $F(p,q)=\sum_x\sqrt{p(x)}\cdot\sqrt{q(x)}$.

    A quantum state $\rho$ in Hilbert space $\mathcal{H}$ is a trace-one positive semi-definite operator acting on $\mathcal{H}$. A quantum state $\rho$ is called pure if it is rank-one as a matrix, namely $\rho=|\psi\rangle\langle\psi|$ for some unit vector $|\psi\rangle$. In this special case, we also write $\rho$ as $\ket{\psi}$.  For a pure state $\ket{\psi} \in \mathcal{H}_A \otimes \mathcal{H}_B$, its Schmidt decomposition is defined as $\ket{\psi}=\sum_i \sqrt{\lambda_i} \ket{\alpha_i} \otimes \ket{\beta_i}$ where $\ket{\alpha_i}$ and $\ket{\beta_i}$ are orthonormal bases for $\mathcal{H}_A$ and $\mathcal{H}_B$ respectively. {The nonnegative real numbers $\sqrt{\lambda_i}$ in the Schmidt decomposition} are called the Schmidt coeﬀicients of $\ket{\psi}$. For quantum states $\rho$ and $\sigma$, the fidelity between them is defined as $\mathrm{F}(\rho, \sigma) = \operatorname{tr}\left(\sqrt{\sigma^{1 / 2} \rho \sigma^{1 / 2}}\right)$.

    \section{The quantum case that $\rho$ is pure}

    We first consider the following special case: If the seed is a bipartite pure state $\ket{\psi}$, can a target classical correlation $P$ be produced from it by local operations only? We now show that this special case already has very rich mathematical structures.

    \subsection{On the computational complexity}

    We first prove that solving this special case generally is already NP-hard, implying that it is hard to solve the original problem efficiently.

    As mentioned in the introduction, for a given classical correlation $P$, its quantum correlation complexity, the minimum size of a quantum seed state $\ket{\psi_0}$ that can generate $P$, has been characterized completely.

    % \begin{Definition}
    % Let $P = [P(x,y)]_{x,y}$ be an $m\times n$ matrix with each entry nonnegative. Then the \red{PSD}-rank of $P$, denoted by $\prank(P) $, is the minimum positive integer $r$ such that there are $r\times r$ PSD matrices $C_x, D_y$ satisfying $P(x,y)=\tr(C_x D_y)$ for any $x\in[m]$ and $y\in[n]$.
    % \end{Definition}

    \begin{lemma}[\cite{jain2013efficient}]\label{lem:qrho}
    Let $P = [P(x,y)]_{x,y}$ be a classical correlation, and $\rho=\sum_{x,y}P(x,y)\cdot\ket{x}\bra{x} \otimes \ket{y}\bra{y}$ be a quantum state in $\h_A \otimes \h_B$. Then it holds that
    \begin{align*}
        \Q(\rho) &= \lceil \log_2 \prank(P) \rceil \\
        &= \min_{\h_{A_1}, \h_{B_1}}\{\big\lceil\log_2\srank(\ket{\psi})\big\rceil: \ket{\psi} \text{ is a pure state}\\
        &\text{ in } \h_{A_1}\otimes \h_A\otimes \h_B \otimes \h_{B_1},  \rho= \tr_{\h_{A_1} \otimes \h_{B_1}} \ket{\psi}\bra{\psi}\},
    \end{align*}
    where $\srank(\ket{\psi})$ is the Schmidt rank of $\ket{\psi}$ with respect to the partition $AA_1|BB_1$.
    \end{lemma}
    That is to say, a given pure state $\ket{\psi_0}$ can generate $P$ only when the size of $\ket{\psi_0}$ is at least $\lceil \log_2 \prank(P) \rceil$. In fact, \cite{jain2013efficient} proves a further conclusion, that is, $\ket{\psi_0}$ can generate $P$ if and only if $\ket{\psi_0}$ has the same Schmidt coefficients with some purification of $\rho$ in $\h_{A_1}\otimes \h_A\otimes \h_B \otimes \h_{B_1}$, where the partition is chosen to be $AA_1|BB_1$.

    This can be explained as follows. On the one hand, if $\ket{\psi_0}$ has the same Schmidt coefficients with a purification $\ket{\phi}$ of $\rho$, then Alice and Bob can transform $\ket{\psi_0}$ to $\ket{\phi}$ by local unitaries (attach some blank ancilla qubits if needed). Then by throwing away the qubits in $\h_{A_1}$ and $\h_{B_1}$, the obtained quantum state will be exactly $\rho$. On the other hand, if $\ket{\psi_0}$ can generate $P$, it means that Alice and Bob can produce $\rho$ by performing local quantum operations on $\ket{\psi_0}$. Note that these local operations can be simulated by attaching necessary ancilla qubits, performing local unitary operations, and then dropping part of the qubits, which indicates that before dropping part of the qubits, the overall quantum state is actually a purification of $\rho$, and it has the same Schmidt coefficients with $\ket{\psi_0}$, as local unitaries do not change them.

    Unfortunately, {we now prove that} it is {generally} NP-hard to determine whether or not $\ket{\psi}$ has the same Schmidt coefficients will some purification of $\rho$ in $\h_{A_1}\otimes \h_A\otimes \h_B \otimes \h_{B_1}$.

	\begin{theorem}
		\label{thm:np-hard-quantum}
		The problem of deciding whether a given pure state $\ket{\psi}$ can generate a given correlation $P$ is $\mathbf{NP}$-hard.
	\end{theorem}
	\begin{proof}
	We first recall the fact that the $\mathsf{SUBSET\mhyphen SUM}$ problem is NP-{hard} {\cite{hartmanis1982computers}}. Suppose we have a set of positive {integers} $S=\{a_1,a_2,...,a_{{r}}\}$, and $T=\sum_ia_i/2$. The $\mathsf{SUBSET\mhyphen SUM}$ problem is asking whether any subset of $S$ sums to precisely $T$. We now try to reduce the $\mathsf{SUBSET\mhyphen SUM}$ problem to the problem we are studying, {denoted as $\mathsf{CORRELATION\mhyphen GENERATION}$ problem for convenience}.
    {We specify the inputs of $\mathsf{CORRELATION\mhyphen GENERATION}$ to be the squared Schmidt coefficients of the pure seed state instead of its amplitudes in the computational basis;
    this does not sacrifice generality as any two pure states with the same Schmidt coefficients are equivalent to each other under local isometries.}

Consider the following classical correlation:
		\begin{align*}
			P=\frac{1}{2}I_{2\times 2}=\frac{1}{2}
			\begin{pmatrix}
				1 & 0 \\
				0 & 1 \\
			\end{pmatrix}.
		\end{align*}
		Let the Schmidt decomposition of $\ket{\psi}$ be
		\begin{align*}
			\ket{\psi}= & \sum_{i=1}^{r}\sqrt{\lambda_{i}}\ket{\alpha_i}\ket{\beta_i},
		\end{align*}
		{where $\lambda_i = \frac{a_i}{\sum_ka_k}$ for any $i$.} Also let $\Lambda=\mathrm{diag}(\sqrt{\lambda_1},\cdots,\sqrt{\lambda_r})$.
        {We remark that a finite representation of $\{\lambda_1,\cdots, \lambda_r\}$ can be obtained from $\{a_1,\cdots, a_r\}$ in polynomial steps due to each $\lambda_i$ being a rational number.}
		Without loss of generality, we assume that $\lambda_1\geq\lambda_2\geq...\geq\lambda_r>0$.
		Then $\ket{\psi}$ can generate correlation $P$ if and only if Alice and Bob can perform binary-outcome quantum measurements with
        operators $\{A,I-A\}$ and $\{B^T,I-B^T\}$ respectively such that it holds that
		\begin{align*}
			\mathrm{tr}(\ket{\psi}\bra{\psi}(A\otimes B^T))=\mathrm{tr}(\Lambda A\Lambda B)= & 1/2,\\
			\mathrm{tr}(\Lambda (I-A)\Lambda (I-B))= & 1/2, \\
			\mathrm{tr}(\Lambda A\Lambda (I-B))= & 0,\\
            \mathrm{tr}(\Lambda (I-A)\Lambda B)= & 0,
		\end{align*}
        where $A$ and $B$ are $r\times r$ PSD matrices satisfying $0\leq A,B\leq I$, and $B^T$ is the transpose of $B$.

        Note that $\mathrm{tr}(\Lambda A\Lambda (I-B))= 0$ implies that
		\begin{align*}
			\mathrm{tr}{(}(\Lambda^{1/2}A\Lambda^{1/2})(\Lambda^{1/2}(I-B)\Lambda^{1/2}){)}= & 0, \\
			\mathrm{tr}{(}(\Lambda^{1/2}(I-B)\Lambda^{1/2})(\Lambda^{1/2}A\Lambda^{1/2}){)}= & 0.
		\end{align*}
		Since both $\Lambda^{1/2}A\Lambda^{1/2}$ and $\Lambda^{1/2}(I-B)\Lambda^{1/2}$ are PSD matrices,
		we immediately have that
		\begin{align*}
			(\Lambda^{1/2}A\Lambda^{1/2})(\Lambda^{1/2}(I-B)\Lambda^{1/2})= & 0, \\
			(\Lambda^{1/2}(I-B)\Lambda^{1/2})(\Lambda^{1/2}A\Lambda^{1/2})= & 0.
		\end{align*}

        By similar argument, $\mathrm{tr}(\Lambda (I-A)\Lambda B)=0$ implies that
		\begin{align*}
			(\Lambda^{1/2}(I-A)\Lambda^{1/2})(\Lambda^{1/2}B\Lambda^{1/2})= & 0, \\
			(\Lambda^{1/2}B\Lambda^{1/2})(\Lambda^{1/2}(I-A)\Lambda^{1/2})= & 0.
		\end{align*}
		Since $\Lambda$ and $\Lambda^{1/2}$ are invertible, the above conditions indicate that
		\begin{align*}
			A\Lambda B= & A\Lambda = \Lambda B, \\
			B\Lambda A= & B\Lambda = \Lambda A.
		\end{align*}
		This gives $B=\Lambda^{-1} A\Lambda=\Lambda A \Lambda^{-1}$.
		Since $\Lambda$ is diagonal, we have $B=A$ and when $\lambda_i\neq \lambda_j$ it
        holds that $A_{ij}=0$, which means that $A$ and $B$ are block diagonal matrices and each block corresponds to each distinct value of $\lambda_i$.
		
        Furthermore, we also have
		\begin{align*}
			B=\Lambda^{-1} A\Lambda=\Lambda^{-1} A\Lambda B=\Lambda^{-1} A^2 \Lambda=B^2,
		\end{align*}
		hence both $A$ and $B$ are projectors. Combining this fact and that $A$ and $B$ are block diagonal matrices, we
        know that each block of $A$ and $B$ are also projectors.

		Finally, we obtain that
		\begin{align*}
			\frac{1}{2}&=\mathrm{tr}(\Lambda A\Lambda B)=\mathrm{tr}(\Lambda^2 A^2)=\mathrm{tr}(\Lambda^2 A)\\
            &=\sum_{\lambda}\lambda\cdot \mathrm{rank}(A|_{\mathrm{ker}(\Lambda^2-\lambda I)}),
		\end{align*}
		where $\lambda$ in the summation runs through all distinct values of $\lambda_i$, and $A|_{\mathrm{ker}(\Lambda^2-\lambda I)}$ is the block of $A$ that corresponds to the value of $\lambda$.
		It can be seen that this summation is exactly a {subset sum} of $\{\lambda_1,\cdots, \lambda_r\}$, {hence $\left(\sum_ia_i\right)\left(\sum_{\lambda}\lambda\cdot \mathrm{rank}(A|_{\mathrm{ker}(\Lambda^2-\lambda I)})\right)$ is a subset sum of $\{a_1,\cdots, a_r\}$.
		Therefore} if $\ket{\psi}$ can generate $P$, {a solution of the corresponding $\mathsf{SUBSET\mhyphen SUM}$ instance also exists}.  Conversely, if the $\mathsf{SUBSET\mhyphen SUM}$ instance has a solution, then by measuring $\ket{\psi}$ in its Schmidt basis and grouping the outcomes,
		we can transform $\ket{\psi}$ into $P$.
		Hence we obtain a {polynomial-time} reduction of the $\mathsf{SUBSET\mhyphen SUM}$ problem to the {$\mathsf{CORRELATION\mhyphen GENERATION}$ problem.}
	\end{proof}
%		Since it is straightforward to find $A$ and $B$ given a solution of the $\mathsf{SUBSET\mhyphen SUM}$ instance,
%		i.e., a set $S\subseteq[n]$ of indices satisfying $\sum_{i\in S}\lambda_i=\frac{1}{2}$

%	\red{[to-do: revise for notation consistency]}

    \subsection{The {diagonal} form of PSD {factorizations}}

    For a given classical correlation $P=[P(x,y)]_{x,y}$ and the corresponding quantum state $\rho=\sum_{x,y}P(x,y)\cdot\ket{x}\bra{x} \otimes \ket{y}\bra{y}$ in $\h_A \otimes \h_B$, we now know that generally it is hard to determine whether a seed quantum state $\ket{\psi_0}$ has the same Schmidt coefficients with some purification of $\rho$ in $\h_{A_1}\otimes \h_A\otimes \h_B \otimes \h_{B_1}$. Here the major challenge is to characterize the Schmidt coefficients for all possible purifications of $\rho$ in $\h_{A_1}\otimes \h_A\otimes \h_B \otimes \h_{B_1}$, and any results of this kind will be very useful to determine whether a given $\ket{\psi}$ can produce $P$ by local quantum operations only.

    Interestingly, it turns out that there exists a close relation between Schmidt coefficients of purifications of $\rho$ in $\h_{A_1}\otimes \h_A\otimes \h_B \otimes \h_{B_1}$ and a special form of PSD factorizations for $P$. Recall that the concept of PSD factorization {plays} an important role in optimization theory, which determines the computational power of semi-definite programming in combinatorial optimization problems {\cite{yannakakis1988expressing, blekherman2012semidefinite}}.

    %    For a given classical correlation $P=[P(x,y)]_{x,y}$ and the corresponding quantum state $\rho=\sum_{x,y}P(x,y)\cdot\ket{x}\bra{x} \otimes \ket{y}\bra{y}$ in $\h_A \otimes \h_B$, we now show that there exists a close relation between purifications of $\rho$ in $\h_{A_1}\otimes \h_A\otimes \h_B \otimes \h_{B_1}$ and PSD-factorizations of $P$. W
    Specifically, we call this special form of PSD factorizations \emph{the {diagonal} form of PSD factorizations}.
    \begin{Definition} A \emph{{diagonal} form of PSD {factorizations}} for a nonnegative matrix $P\in \mbR_+^{n\times m}$ is a collection of PSD matrices $C_x,D_y\in \mbC^{k\times k}$ that satisfy
    \[P(x,y)=\tr (C_xD_y),\ x=1,...,n,\ y=1...,m\]
    and
    \[\sum_{x=1}^nC_x=\sum_{y=1}^mD_y=\Lambda,\]
    where $\Lambda\in \mbC^{k\times k}$ is a diagonal nonnegative matrix. And if $k=\prank(P)$, we say this is an optimal {diagonal} form of PSD {factorizations}.
    \end{Definition}
    We remark that for an arbitrary nonnegative matrix $M$, one can always find an optimal {diagonal} form of PSD {factorizations}, which is implied by the following theorem. Note that to make the entries of $M$ sum to $1$, a proper renormalization for $M$ by a constant factor may be needed.
    % Furthermore, our numerical calculations show that a nonnegative matrix can have multiple optimal {diagonal} form of PSD factorization.

    % \begin{theorem}\label{thm:canonical} Let $P=[P(x,y)]_{x,y}$ be a classical correlation, and $\rho=\sum_{x,y}P(x,y)\cdot\ket{x}\bra{x} \otimes \ket{y}\bra{y}\in\h_A \otimes \h_B$. Suppose $\ket{\psi}$ is a purification of $\rho$ in $\h_{A_1}\otimes \h_A\otimes \h_B \otimes \h_{B_1}$, then {there exists a {diagonal} form of PSD factorization $\{C_x,D_y\}$ of $P$ such that the Schmidt coefficients of $\ket{\psi}$ with respect to the partition $AA_1|BB_1$ are exactly the diagonal entries of $\Lambda=\sum_{x=1}^{n}C_x=\sum_{y=1}^{m}D_y$.}
    % \end{theorem}
    \begin{lemma}\label{thm:canonical} Let $P=[P(x,y)]_{x,y}$ be a classical correlation, and $\rho=\sum_{x,y}P(x,y)\cdot\ket{x}\bra{x} \otimes \ket{y}\bra{y}\in\h_A \otimes \h_B$. Then there exists a purification $\ket{\psi}$ of $\rho$ with Schmidt coefficients
    $\sqrt{\lambda_1},\sqrt{\lambda_2},...\sqrt{\lambda_r}$ if and only if
    there exists a {diagonal} form of PSD factorization $\{C_x,D_y\}$ of $P$ such that $\sum_{x=1}^{n}C_x=\sum_{y=1}^{m}D_y = \mathrm{diag}(\sqrt{\lambda_1},\cdots,\sqrt{\lambda_r})$.
    \end{lemma}
    \begin{proof} Suppose the Schmidt decomposition {of $\ket{\psi}\in\h_{A}\otimes \h_{A_1}\otimes \h_B \otimes \h_{B_1} $}  can be written as
    \[
    \ket{\psi} =  \sum_{i=1}^r \left(\sum_{x=1}^n \ket{x} \otimes \ket{v_x^i}\right) \otimes \left(\sum_{y=1}^m \ket{y} \otimes \ket{w_y^i}\right),
    \]
    where $\ket{v_x^i}\in\h_{A_1}$ and $\ket{w_y^i}\in\h_{B_1}$ are unnormalized quantum states. According to the definition of Schmidt decomposition, one can choose the norms of $\ket{v_x^i}\in\h_{A_1}$ and $\ket{w_y^i}\in\h_{B_1}$ properly such that
    \begin{equation}\label{eq:canonical1}
    \sum_{x=1}^n\qip{v_{x}^j}{v_x^i}=\sum_{y=1}^m\qip{w_{y}^j}{w_y^i}=0, \ \ \text{for any }i\neq j,
    \end{equation}
    and
    \begin{equation}\label{eq:canonical2}
    \sum_{x=1}^n\qip{v_{x}^i}{v_x^i}=\sum_{y=1}^m\qip{w_{y}^i}{w_y^i}=\sqrt{\lambda_i}, \ \ \text{for any }i\in[r].
    \end{equation}

    Since $\ket{\psi}$ is a purification of $\rho$, it holds that
    \begin{align*}
     \rho & = \tr_{\h_{A_1} \otimes \h_{B_1}} \ket{\psi} \bra{\psi}  \\
    & =  \sum_{x, y} \ket{x}\bra{x} \otimes \ket{y}\bra{y} \left( \sum_{i,j=1}^r \qip{v_{x}^j}{v_x^i} \cdot \qip{w_{y}^j}{w_y^i}\right).
    \end{align*}
    We now define $r \times r$ matrices $C_x$ such that $C_x(j,i) = \qip{v_x^j}{v_x^i}$ for all $i,j \in [r]$, and $r \times r$ matrices $D_y$ such that $D_y(i,j) = \qip{w_y^j}{w_y^i}$ for all $i,j \in [r]$. Then for any $x\in[n]$ and $y\in[m]$, $C_x$ and $D_y$ are positive semi-definite matrices, and Eqs.\eqref{eq:canonical1} and \eqref{eq:canonical2} imply that
    \[
    \sum_{x=1}^nC_x=\sum_{y=1}^mD_y=\Lambda, \ \ \text{for any }i\in[r],
    \]
    where $\Lambda=\text{diag}(\sqrt{\lambda_1},...,\sqrt{\lambda_r})$.

    For the other direction, let $C_x, D_y \in \mathbb{C}^{r \times r}$ be positive semidefinite matrices with $\operatorname{tr}\left(C_x D_y\right)=P(x, y)$ {for all} $x\in[n],\ y \in [m]$, and this decomposition {satisfy} $\sum_{x=1}^{n}C_x=\sum_{y=1}^{m}D_y = \mathrm{diag}(\sqrt{\lambda_1},\cdots,\sqrt{\lambda_r})$. For $i \in[r]$, let $\left|v_x^i\right\rangle$ be the $i$-th column of $\sqrt{C_x^T}$ and let $\left|w_y^i\right\rangle$ be the $i$-th column of $\sqrt{D_y}$. Define $|\psi\rangle$ in $\mathcal{H}_A \otimes \mathcal{H}_A \otimes \mathcal{H}_{A_1} \otimes \mathcal{H}_B \otimes \mathcal{H}_B \otimes \mathcal{H}_{B_1}$ as follows:
    $$|\psi\rangle \!\stackrel{\text { def }}{=}\! \sum_{i=1}^r\!\left(\!\sum_x|x\rangle \!\otimes\!|x\rangle \!\otimes\!\left|v_x^i\right\rangle\!\right) \!\otimes\!\left(\!\sum_y|y\rangle \!\otimes\!|y\rangle \!\otimes\!\left|w_y^i\right\rangle\!\right)\!.$$
    {Then, by defining $\ket{V_i} = \sum_x|x\rangle \otimes|x\rangle \otimes\left|v_x^i\right\rangle$ and $\ket{W_i} = \sum_y|y\rangle \otimes|y\rangle \otimes\left|w_y^i\right\rangle$, it can be shown directly that}
    \begin{equation*}
        \bra{V_{i'}}V_{i}\rangle = \sum_x C_x^T(i',i) = 0, \ \ \text{for any }i\neq i' \in [r],
        \end{equation*}
    and
    \begin{equation*}
        \bra{V_{i}}V_{i}\rangle = \sum_x C_x^T(i,i) = \sqrt{\lambda_i}, \ \ \text{for any }i\in[r].
    \end{equation*}
    Similarly, it holds that $\bra{W_{i'}}W_{i}\rangle = \sum_y D_y(i',i) = 0$ and $\bra{W_{i}}W_{i}\rangle = \sum_y D_y(i,i) = \sqrt{\lambda_i}$ {for all $i\neq i' \in [r]$.} The above four equations imply $\ket{\psi}$ is a pure state with Schmidt coefficients
    $\sqrt{\lambda_1},\sqrt{\lambda_2},...\sqrt{\lambda_r}$ with respect to the partition $AAA_1|BBB_1$, also
    \begin{align*}
        & \operatorname{tr}_{\mathcal{H}_A \otimes \mathcal{H}_{A_1} \otimes \mathcal{H}_B \otimes \mathcal{H}_{B_1}}{|\psi\rangle\langle\psi|} \\
        = & \sum_{x, y}|x\rangle\langle x|\otimes| y\rangle\langle y|\left(\sum_{i, j=1}^r\left\langle v_x^j \mid v_x^i\right\rangle \cdot\left\langle w_y^j \mid w_y^i\right\rangle\right) \\
        = & \sum_{x, y}|x\rangle\langle x|\otimes| y\rangle\langle y| \cdot \operatorname{tr}\left(C_x D_y\right)=\rho,
        \end{align*}
    which shows $\ket{\psi}$ is a purification of $\rho$.

    \end{proof}

    {Since {Lemma} \ref{thm:canonical} gives a complete {description for the Schmidt coefficients} of all possible purifications of $\rho=\sum_{x,y}P(x,y)\cdot\ket{x}\bra{x} \otimes \ket{y}\bra{y}$, we directly obtain the following Theorem that determines whether a pure state $\ket{\psi}$ can be transformed to $P$ under local transformation.}

    \begin{theorem}\label{thm:transform condition}
    {A pure state $\ket{\psi}$ with Schmidt coefficients $\sqrt{\lambda_1},\sqrt{\lambda_2},...\sqrt{\lambda_r}$ can generate a classical correlation} $P=[P(x,y)]_{x,y}$ if and only if there exists a diagonal form of PSD factorization $\{C_x,D_y\}$ of $P$ such that $\sum_{x=1}^{n}C_x=\sum_{y=1}^{m}D_y = \mathrm{diag}(\sqrt{\lambda_1},\cdots,\sqrt{\lambda_r})$.
    \end{theorem}

    An interesting fact is that a nonnegative matrix can have multiple optimal {diagonal} {forms of PSD factorizations}. For example, let
    \[P = \frac{1}{3}\begin{bmatrix}
    1 & 1\\
    1 & 0
    \end{bmatrix},\]
    then it holds that $\prank(P)=2$. We define two pure states in $\h_{A}\otimes \h_B\otimes \h_{A_1} \otimes \h_{B_1}$ as
    \[\ket{\psi_1}={\sum_{(x,y)\in \{0,1\}^2}}\sqrt{P(x,y)}\ket{x}\ket{y}\ket{x}\ket{y}\]
    and
    \[{\ket{\psi_2}}={\sum_{(x,y)\in \{0,1\}^2}}\sqrt{P(x,y)}\ket{x}\ket{y}{\rm CNOT}\ket{x}\ket{y},\]
    where the CNOT gate is performed on the subsystems $A_1$ and $B_1$. Then it can be verified that $\ket{\psi_1}$ and $\ket{\psi_2}$ are two different purifications of $\rho=\sum_{x,y}P(x,y)\cdot\ket{x}\bra{x} \otimes \ket{y}\bra{y}$ in $\h_{A_1}\otimes \h_A\otimes \h_B \otimes \h_{B_1}$. By calculating their Schmidt coefficients, we can obtain two different optimal {diagonal} forms of PSD {factorizations} for $P$, where the $\Lambda$ matrices are $\Lambda_1\approx\text{diag}(0.9342,0.3568)$ and $\Lambda_2\approx\text{diag}(0.8165, 0.5774)$ respectively.

%    of $\rho=\sum_{x,y}P(x,y)\cdot\ket{x_A}\bra{x_A} \otimes \ket{y_B}\bra{y_B}$, where $\ket{x_Ay_B}$ and $\ket{x_{A'}y_{B'}}$ are computational bases for the system $AB$ and the reference system $A'B'$ in $\mathcal{H}^2\otimes\mathcal{H}^2$. By calculating the Schmidt coefficients between systems $AA'$ and $BB'$ of $\ket{\psi_1}$ and $\ket{\psi_2}$, we obtain two different optimal {diagonal} forms of $P$, say $\Lambda_1=\text{diag}(0.9342,0.3568)$ and $\Lambda_2=\text{diag}(0.8165, 0.5774)$, according to Theorem \ref{thm:{diagonal}}.

    \subsection{Several necessary conditions for that $\ket{\psi}$ generates $P$}

    Let $m(\cdot)$ be a measure of bipartite states which is monotone non-increasing under local quantum operations. Then it is easy to see that $m(\ket{\psi}\bra{\psi})\geq m(P)$ is a necessary condition for which $\ket{\psi}$ can be locally transformed to the target classical correlation $P$.
    Several information-theoretic quantities are known to be monotone non-increasing under local operations such as mutual information. Thus, if the mutual information $I(A:B)$ of $\ket{\psi}$ is less than that of $P$, then $P$ cannot be generated from $\ket{\psi}$ under local operations only.

    Another such measure of correlation is the \emph{sandwiched ${\alpha\mhyphen}$R\'{e}nyi divergence} \cite{muller2013quantum}, which is defined as
    \[\tilde{D}_{\alpha}(\rho\Vert\sigma)=\frac{1}{\alpha-1}\log\left(\frac{1}{\tr(\rho)}\tr\left[(\sigma^{\frac{1-\alpha}{2\alpha}}\rho\sigma^{\frac{1-\alpha}{2\alpha}})^{\alpha}\right]\right),\]
    where $\alpha\in(0,1)\cup(1,\infty]$. Let $\ket{\psi}=\sum_i\sqrt{\lambda_i}\ket{i_A}\ket{i_B},\rho=\ket{\psi}\bra{\psi}$ and $\sigma=\sum_{x,y}P(x,y)\cdot\ket{x}\bra{x} \otimes \ket{y}\bra{y}$, where $\ket{i_A},\ket{x}\in\mathcal{H}_A$ ($\ket{i_B},\ket{y}\in\mathcal{H}_B$) are orthonormal bases for systems $A$ ($B$). Since the data processing inequality holds for the order-$\alpha$ R\'{e}nyi divergence with $\alpha\in[\frac{1}{2},1)\cup(1,\infty]$, i.e., it cannot be increased by local quantum operations, by requiring $\tilde{D}_{\alpha}(\rho\Vert\rho_A\otimes\rho_B)\ge\tilde{D}_{\alpha}(\sigma\Vert\sigma_A\otimes\sigma_B)$ where $\rho_A,\rho_B,\sigma_A,\sigma_B$ are reduced density matrices for $\rho$ and $\sigma$ respectively, it derives
    \begin{equation}\label{eq:renyi}
        \left\{\begin{array}{cc}
        \left(\sum_i\lambda_i^{\frac{2}{\alpha}-1}\right)^{\alpha}\le \sum_{xy}\frac{P(x,y)^{\alpha}}{(P(x)P(y))^{\alpha-1}}, & \alpha\in[\frac{1}{2},1)\\
        \left(\sum_i\lambda_i^{\frac{2}{\alpha}-1}\right)^{\alpha}\ge\sum_{xy}\frac{P(x,y)^{\alpha}}{(P(x)P(y))^{\alpha-1}}, & \alpha\in(1,\infty)\\
        {\sum_i \frac{1}{\lambda_i} \geq \max\limits_{x,y} \frac{P(x,y)}{\sum_{j}P(x,{j})\sum_{i}P({i},y)},} & {\alpha=\infty.}\\
        \end{array}
        \right.
    \end{equation}
    Therefore, if $\ket{\psi}$ and $P$ violate any inequality in Eq.(\ref{eq:renyi}), $\ket{\psi}$ cannot generate $P$ under local operations.

    In addition to the above natural conclusions, we now show that the insight provided by Theorem \ref{thm:canonical} allows us to obtain several new necessary conditions that $\ket{\psi}$ has to satisfy to generate $P=[P(x,y)]_{x,y}$, which can be stronger than the above two necessary conditions. Again, we suppose the Schmidt rank of $\ket{\psi}$ is $r$, and the squared Schmidt coefficients are $\lambda_1\geq\lambda_2\geq...\geq\lambda_r$. For simplicity, in this subsection we suppose that $r=\prank(P)$, i.e., $\lambda_r>0$.

    % \subsubsection{}\label{sec:minimal}
    \begin{proposition}(\text{A necessary condition for $\lambda_r$})\label{condition:first}
        If $\ket{\psi}$ can produce $P=[P(x,y)]_{x,y}$ by local operations only, then it holds that  \begin{equation}\label{eq:first_condition}
        \lambda_r\leq\min_{x,y}\frac{\sum_{{j}}P(x,{{j}})\sum_{{i}}P({{i}},y)}{P(x,y)}.
        \end{equation}
        \end{proposition}
    \begin{proof}
    As mentioned before, $\ket{\psi}$ can generate $P$ {if and only if} $\ket{\psi}$ has the same Schmidt coefficients with a purification of $\rho$. According to Theorem \ref{thm:canonical}, this means that $P$ has a {diagonal} form of PSD {factorizations} $C_{{i}}$ and $D_{{j}}$ with $\sum_{{i}}C_{{i}}=\sum_{{j}}D_{{j}}=\Lambda=\text{diag}(\sqrt{\lambda_1},...,\sqrt{\lambda_r})$, {where $i\in[n]$ and $j\in[m]$. For arbitrary $x\in[n]$ and $y\in[m]$}, we have that
    $$
    \sum_{{j}}P(x,{j})=\sum_{{j}}\tr (C_xD_{{j}})=\tr(C_x\cdot\Lambda).
    $$
    Since $\Lambda\geq \sqrt{\lambda_r}\cdot I_{r\times r}$, it holds that $\sum_{{j}}P(x,{j})\geq\sqrt{\lambda_r}\tr(C_x)$, where $I_{r\times r}$ is the $r\times r$ identity matrix. Similarly, we also have $\sum_{{i}}P({i},y)\geq\sqrt{\lambda_r}\tr(D_y)$.

    Meanwhile, according to the Cauchy-Schwarz inequality, the relation $P(x,y)=\tr(C_xD_y)$ implies that
    \[
    P(x,y)\leq\sqrt{\tr(C_x^2)}\sqrt{\tr(D_y^2)}\leq\tr(C_x)\tr(D_y),
    \]
    where we have utilized the fact that for any PSD matrix $M$ it holds that $(\tr(M))^2\geq\tr(M^2)$. Combining the above {facts}, we have that
    $$
    \sum_{{j}}P(x,{j})\sum_{{i}}P({i},y)\geq\lambda_r\tr(C_x)\tr(D_y)\geq\lambda_r\cdot P(x,y).
    $$
    Therefore, we conclude the proof.
    \end{proof}
    % \begin{proposition}\label{condition:first}
    % If $\ket{\psi}$ can produce $P=[P(x,y)]_{x,y}$ by local operations only, then it holds that  \begin{equation}\label{eq:first_condition}
    % \lambda_r\leq\min_{x,y}\frac{\sum_{\red{j}}P(x,{\red{j}})\sum_{\red{i}}P({\red{i}},y)}{P(x,y)}.
    % \end{equation}
    % \end{proposition}

    Interestingly, Proposition \ref{condition:first} can be compared with the result given by the \emph{sandwiched $\alpha\mhyphen$R\'{e}nyi divergence}. Recall that when $\alpha=\infty$, Eq.(\ref{eq:renyi}) yields
    $$\sum_i \frac{1}{\lambda_i} \geq \max\limits_{x,y} \frac{P(x,y)}{\sum_{j}P(x,{j})\sum_{i}P({i},y)},$$
    which is always weaker than Eq.(\ref{eq:first_condition}).
    \begin{example}
    Let us see an example showing that this condition can be stronger than the ones given by the mutual information and the sandwiched ${\alpha\mhyphen}$R\'{e}nyi divergence. Suppose one tries to generate $P=\begin{bmatrix}
    0.3 & 0\\
    0 & 0.7
    \end{bmatrix}$ from $\ket{\psi}=\frac{\ket{00}+\ket{11}}{\sqrt{2}}$ under local operations. Firstly, it can be verified that $P$ and $\ket{\psi}$ satisfy the conditions given by the mutual information condition and the sandwiched ${\alpha\mhyphen}$R\'{e}nyi divergence. However, note that the minimal squared Schmidt coefficient $\lambda_r$ of $\ket{\psi}$ is 0.5, while the {right-hand side} of Eq.(\ref{eq:first_condition}) for $P$ is 0.3, implying that Eq.(\ref{eq:first_condition}) is violated and the possibility of producing $P$ by performing local operations on $\ket{\psi}$ is ruled out.
    \end{example}

    % \subsubsection{A necessary condition for the entropy of $\{\lambda_i\}$}\label{sec:entropy}

    \begin{proposition}(A necessary condition for the entropy \\
        of $\{\lambda_i\}$)\label{condition:second}
        If $\ket{\psi}$ can produce $P=[P(x,y)]_{x,y}$ by local operations only, then it holds that  \begin{equation}\label{eq:Holevo_condition}
        I(P)\leq -\sum_i\lambda_i\log(\lambda_i).
        \end{equation}
        \end{proposition}
        Before we prove this proposition, it is worth mentioning that the necessary condition provided by the monotonicity of mutual information directly reads
        \[
        I(P)\leq -2\sum_i\lambda_i\log(\lambda_i),
        \]
        indicating that the above new condition is always stronger.
    \begin{proof}
    Define $V=\text{diag}(\sqrt[4]{\lambda_1},\sqrt[4]{\lambda_2},...,\sqrt[4]{\lambda_r})$. Since $C_x$ and $D_y$ are a {diagonal} form of PSD {factorizations} for $P$, it holds that $C_x'=VC_xV$ and $D_y'=V^{-1}D_yV^{-1}$ are still a valid PSD decomposition for $P$, i.e., $P(x,y)=\tr(C_x'D_y')$. Furthermore, we now have that $\sum_xC_x'=\text{diag}(\lambda_1,\lambda_2,...,\lambda_r)$ and $\sum_yD_y'= I_{r\times r}$.

    At the same time, let $P_x$ denote the transpose of the $x$-th row of $P$, where $x\in[n]$. Define $\text{sum}(P_x)$ to be the {sum} of all the entries of $P_x$, then it holds that $\text{sum}(P_x)=\sum_y\tr(C_x'D_y')=\tr(C_x')$. Since $C_x'$ is PSD for any $x$ and $\sum_yD_y'=I_{r\times r}$, we now view $\{D_y'\}$ as a  set of a positive-operator-valued measure (POVM), and $\rho_{x}=C_x'/\tr(C_x')$ as a quantum state, then $P_x/\tr(C_x')$ is actually the outcome probability distribution if we measure $\rho_{x}$ with the POVM $\{D_y'\}$.

    We now consider the following virtual protocol performed by Alice and Bob. With probability $\tr(C_x')$, Alice prepares and sends the quantum state $\rho_{x}$ to Bob. After receiving the quantum state, he measures it using the POVM $\{D_y'\}$ and records the outcome $y$. This process can be regarded as a protocol that classical information is transformed from Alice to Bob. According to the Holevo bound {(\cite{10.5555/1972505}, chapter 12)}, it holds that
    \begin{align*}
    I(P)&\leq S\left(\sum_x\tr(C_x')\cdot\rho_{x}\right)-\sum_x\tr(C_x')\cdot S(\rho_{x})\\
    &\leq S\left(\sum_xC_x'\right)=-\sum_i\lambda_i\log(\lambda_i),
    \end{align*}
    where $I(P)$ is the mutual information between $x$ and $y$, $S(\rho)$ is the von Neumann entropy of $\rho$, and we have utilized the fact that $S(\rho_{y})\geq0$ for any $y$.
    \end{proof}
    % \begin{proposition}\label{condition:second}
    % If $\ket{\psi}$ can produce $P=[P(x,y)]_{x,y}$ by local operations only, then it holds that  \begin{equation}\label{eq:Holevo_condition}
    % I(P)\leq -\sum_i\lambda_i\log(\lambda_i).
    % \end{equation}
    % \end{proposition}
    % Note that the necessary condition provided by the monotonicity of mutual information directly reads
    % \[
    % I(P)\leq -2\sum_i\lambda_i\log(\lambda_i),
    % \]
    % indicating that the above new condition is always stronger.
    Furthermore, the new condition can also be stronger than the necessary condition given by Eq.(\ref{eq:first_condition}) in some cases.
    \begin{example}
    Consider generating the classical correlation  $P=\frac{1}{9}\begin{bmatrix}
        1 & 4\\
        4 & 0
        \end{bmatrix}$ under local operation from the 2-qubit seed state with Schmidt coefficients $1/3$ and $2\sqrt{2}/3$, it is easy to verify that Eq.(\ref{eq:first_condition}) is satisfied.  However, by checking the Eq.(\ref{eq:Holevo_condition}), we have $I \approx 0.59 \leq -\sum_i\lambda_i\log(\lambda_i) \approx 0.5033$, excluding the possibility that generates $P$ from this seed state.
    \end{example}

    % \subsubsection{The first necessary condition for the sum of squares of $\{\lambda_i\}$}
    We now demonstrate how the diagonal form of PSD factorizations can facilitate the connection of our problem with a new correlation measure $V_\alpha(A ; B)$ for bipartite quantum states, which is defined as \cite{mojahedian2019correlation}
    \begin{equation}\label{eq:corr_measure}
        \begin{aligned}
            V_\alpha&(A ; B)\\
            &=\Bigl\|\left(I_B \otimes \rho_A^{-(\alpha-1) / 2 \alpha}\right) \rho_{B A}\left(I_B \otimes \rho_A^{-(\alpha-1) /(2 \alpha)}\right)\\
            & \qquad \qquad \qquad \qquad \qquad \qquad \qquad  -\rho_B \otimes \rho_A^{1 / \alpha} \Bigl\|_{(1, \alpha)},
            \end{aligned}
    \end{equation}
    where $\alpha\in[1,\infty)$, $\rho_A$ and $\rho_B$ are the reduced density matrices for Alice and Bob respectively, and $\|\cdot\|_{(1, \alpha)}$ is the $(1, \alpha)$ norm derived from the interpolation theory \cite{pisier1998non}. For any {$0 \leq p \leq q$}  and $M_{A B} \in \h_A\otimes \h_B, \left\|M_{A B}\right\|_{(p, q)}$ is defined as
    $$
    \begin{aligned}
    \left\|M_{A B}\right\|_{(p, q)}
    =\inf _{\sigma_A, \tau_A}\left\|\left(\sigma_A^{-\frac{1}{2 r}} \otimes I_B\right) M_{A B}\left(\tau_A^{-\frac{1}{2 r}} \otimes I_B\right)\right\|_q,
    \end{aligned}
    $$
    where $r \in(0,+\infty]$ is chosen such that $\frac{1}{p}=\frac{1}{q}+\frac{1}{r}$, $\|\cdot\|_p$ is the Schatten $p$-norm, and the infimum is taken over all density matrices $\sigma_A, \tau_A \in \h_A$.
    {In \cite{mojahedian2019correlation}, the properties of the correlation measure $V_\alpha(A ; B)$ have been well studied, where it} has been known that $V_\alpha(A ; B)$ is monotonically non-increasing under local operation, indicating that it can be utilized to quantify the correlation between subsystems A and B. Note that $V_\alpha(A ; B)$ is hard to compute analytically, but for classical correlation $P$, it can be represented as {\cite{mojahedian2019correlation}}
    \begin{equation}\label{eq:corr_classical}
        V'_\alpha(A ; B)=\sum_y\left(\sum_x P(x)|P(y|x)-P(y)|^\alpha\right)^{1 / \alpha}.
    \end{equation}
    Moreover, an equivalent expression for $V_2(A ; B)$ when $\alpha = 2$ has been given~\cite{mojahedian2019correlation}:
    \begin{equation}\label{eq:V_2}
        \begin{aligned}
            V&_2(A ; B)\\
            &=\!\inf _{\tau_B, \sigma_B}\!\bigg(\operatorname{tr}\!\left[\left(\!\rho_A^{-1 / 2} \!\otimes\! \tau_B^{-1 / 2}\!\right) \rho_{A B}\!\left(\!\rho_A^{-1 / 2} \!\otimes \!\sigma_B^{-1 / 2}\!\right)\! \rho_{A B}\right]\\
            &\qquad \qquad \qquad \qquad \qquad  -\operatorname{tr}\left[\tau_B^{-1 / 2} \rho_B \sigma_B^{-1 / 2} \rho_B\right]\bigg)^{1 / 2}, \\
            \end{aligned}
    \end{equation}
    where the infimum is taken over all density matrices $\tau_B, \sigma_B \in \h_B$. With the concepts given above, we prove the following necessary condition
    \begin{proposition}(The first necessary condition for the sum of squares of $\{\lambda_i\}$)\label{condition:third}
        If $\ket{\psi}$ can produce $P=[P(x,y)]_{x,y}$ by local operations only, then it holds that  \begin{equation}\label{eq:V2_condition}
            \sum_i \lambda_i^2 \leq 1 - V'_2(A ; B)^2/r.
        \end{equation}
        \end{proposition}
    \begin{proof}
    Using the same notations $C_x'$, $D_y'$ and $\rho_{x}$ that we just defined {in Proposition \ref{condition:second}}, we construct a quantum state $\rho_{AB}$ shared by Alice and Bob, which is expressed as
    $$\rho_{AB}= \sum_{x} P(x) \ket{x}\bra{x}\otimes \rho_{x}.$$
    It can be seen that $P$ can be generated directly by measuring this state by the POVM $\{\ket{x}\bra{x}\otimes D_y\}_{x,y}$, which implies that $V_2(A ; B) \geq V'_2(A ; B)$. Let $\tau_B= \sigma_B= I/r$ {in Eq.(\ref{eq:V_2})}, then $V_2(A ; B)^2$ can be upper bounded by
    \begin{equation}\label{eq:upper_V_2}
        \begin{aligned}
        V_2(A ; B)^2 & \leq \tr\left(\sum_x P(x)\ket{x}\bra{x}\otimes r \rho_x^2\right) - r\tr(\rho_B^2)\\
        & \leq r\left(\tr\left(\sum_x P(x) \rho_x^2\right)-\sum_i \lambda_i^2\right)\\
        & \leq r\left(\tr\left(\sum_x P(x) \rho_x\right)-\sum_i \lambda_i^2\right)\\
        & = r(1-\sum_i \lambda_i^2).
        \end{aligned}
        \end{equation}
    In the second inequality we have used the fact that $\rho_A = \sum_x P(x) \ket{x}\bra{x}$ and $\rho_B =\sum_xC_x=\Lambda$. According to the relation $V_2(A ; B) \geq V'_2(A ; B)$, we eventually prove the conclusion.
    \end{proof}
    % \begin{proposition}\label{condition:third}
    % If $\ket{\psi}$ can produce $P=[P(x,y)]_{x,y}$ by local operations only, then it holds that  \begin{equation}\label{eq:V2_condition}
    %     \sum_i \lambda_i^2 \leq 1 - V'_2(A ; B)^2/r,
    % \end{equation}
    % where $V'_2(A ; B)$ is given in Eq.\eqref{eq:corr_classical}.
    % \end{proposition}
    Again, we now see an example showing that this condition can be useful.
    \begin{example}
     Suppose when the pure state $\ket{\psi}=\frac{1}{\sqrt{10}}\ket{00}+\frac{3}{\sqrt{10}}\ket{11}$ is given, Alice and Bob aim at generating the correlation     $P=\frac{1}{11}\begin{bmatrix}
        2 & 6\\
        3 & 0
        \end{bmatrix}$ under local operations. The previous necessary condition given in Eq.(\ref{eq:Holevo_condition}) is satisfied in this case, nevertheless this task cannot be fulfilled since it violates Eq.(\ref{eq:V2_condition}): $\sum_i \lambda_i^2 = 0.82 > 1 - V'_2(A ; B)^2/2 \approx 0.7769$.
        \end{example}

    \begin{proposition}(The second necessary condition for the sum of squares of $\{\lambda_i\}$)\label{condition:fourth}
        If $\ket{\psi}$ can produce $P=[P(x,y)]_{x,y}$ by local operations only, then it holds that  \begin{equation}\label{eq:fourth_condition}
        \sum_{i=1}^n\sum_{j=1}^nF(P_i,P_j)^2\geq\sum_i\lambda_i^2.
        \end{equation}
        \end{proposition}
    \begin{proof}
    We continue using the notations $P_x$ and $\rho_x$ defined in {Proposition \ref{condition:second}}. Note that for $x_1,x_2\in[n]$, $P_{x_1}/\tr(\rho_{x_1})$ and $P_{x_2}/\tr(\rho_{x_2})$ can be seen as two probability distributions produced by measuring $\rho_{x_1}$ and $\rho_{x_2}$ with {the} same POVM $\{D_y'\}$. Therefore, we have that
    \begin{align*}
    F(P_{x_1}/\tr(\rho_{x_1}), &P_{x_2}/\tr(\rho_{x_2}))^2\\
    &\geq F(\rho_{x_1},\rho_{x_2})^2\geq\tr(\rho_{x_1}\rho_{x_2}),
    \end{align*}
    where we have used the facts that the fidelity between two quantum states is smaller than that between the two probability distributions produced by measuring them with the same POVM, and that $F(\rho,\sigma)^2\geq\tr(\rho\sigma)$ for any $\rho$ and $\sigma$ {\cite{lee2017some}}. Therefore,
    \[
    F(P_{x_1},P_{x_2})^2\geq\tr(C_{x_1}'C_{x_2}').
    \]
    Combining this with the fact that $\sum_xC_x'=\text{diag}(\lambda_1,\lambda_2,...,\lambda_r)$, we eventually obtain
    \[
    {\sum_{i=1}^n\sum_{j=1}^n}F(P_i,P_j)^2\geq{\sum_{i=1}^n\sum_{j=1}^n}\tr(C_i'C_j')=\sum_i\lambda_i^2.
    \]
    \end{proof}
    % \begin{proposition}\label{condition:fourth}
    % If $\ket{\psi}$ can produce $P=[P(x,y)]_{x,y}$ by local operations only, then it holds that  \begin{equation}\label{eq:fourth_condition}
    % \sum_{i,j=1}^nF(P_i,P_j)^2\geq\sum_i\lambda_i^2.
    % \end{equation}
    % \end{proposition}

    As usual, we demonstrate an example showing that this necessary condition is not covered by the previous one given in Eq.(\ref{eq:V2_condition}).

    \begin{example}
    Consider the possibility of applying local operations on $\ket{\psi}=\frac{2}{5}\ket{00}+\frac{\sqrt{21}}{5}\ket{11}$ to generate a correlation $P=\frac{1}{11}\begin{bmatrix}
        2 & 6\\
        3 & 0
    \end{bmatrix}$. By calculating the value of the terms in Eq.(\ref{eq:fourth_condition}), we have that {${\sum_{i=1}^2\sum_{j=1}^2}F(P_i,P_j)^2=\frac{85}{121}\approx 0.7025$} and $\sum_i\lambda_i^2=\frac{457}{625}= 0.7312$, ruling out the possibility of generating $P$ from $\ket{\psi}$. However, $0.7312 \leq 1 - V'_2(A ; B)^2/2 \approx 0.7769$, satisfying the necessary condition in Eq.(\ref{eq:V2_condition}).
    \end{example}

    Lastly, we clarify that among the aforementioned four new necessary conditions based on the {diagonal} form of PSD decomposition, no one is always stronger than the others, which means that these necessary conditions could be useful in different scenarios. To show that this is indeed the case, we now see an example for which the necessary condition given in Eq.\eqref{eq:fourth_condition} is satisfied, but the necessary condition given in Eq.\eqref{eq:first_condition} is violated. %Although in the preceding part of the text, we show that in some cases, the latter condition can be stronger than the previous one, examples also exist
    %where the first necessary condition performs better than the last one, say, the condition for $\lambda_r$ witnesses the contradiction but the second condition for the sum of squares of $\{\lambda_i\}$ does not.

    \begin{example}
    Specifically, suppose the target correlation $P=\frac{1}{10}\begin{bmatrix}
        4 & 1 & 1\\
        1 & 1 & 0\\
        1 & 0 & 1
    \end{bmatrix}$ with $\operatorname{rank}_{\mathrm{psd}}(P)=2$, then let us consider whether the Bell state can generate $P$ under local operations. In this case, $\sum_i \lambda_i^2 = 0.5$ and ${\sum_{i=1}^3\sum_{j=1}^3}F(P_i,P_j)^2=0.82$, then the necessary condition Eq.(\ref{eq:fourth_condition}) is satisfied. However, by directly computing the upper bound of $\lambda_r$ given by the necessary condition Eq.(\ref{eq:first_condition}), we have that $\min\limits_{x,y}\frac{\sum_{{j}}P(x,{{j}})\sum_{{i}}P({{i}},y)}{P(x,y)} = 0.4 < 0.5$, which implies that the Bell state can not serve as the seed state to generate this classical correlation.
    \end{example}

    \section{The general case that the seed state is mixed}

    \subsection{A general necessary condition for the case of mixed seed}

    We now turn to the general case of the classical correlation generation problem: If Alice and Bob share a mixed seed quantum state $\rho_0$, can they generate $P$ by local operations without communication? According to Theorem \ref{thm:np-hard-quantum}, it is unlikely to efficiently solve this problem.

    In fact, based on the necessary conditions for that a pure seed quantum state can generate $P$, one can also build necessary conditions that $\rho_0$ must satisfy to generate $P$. For this, one can choose a specific purification of $\rho_0$ in $\h_{A_1}\otimes \h_A\otimes \h_B \otimes \h_{B_1}$ \emph{arbitrarily}, denoted $\ket{\psi_0}$, and then determine whether $\ket{\psi_0}$ can generate $P$ by using the results we have obtained. If we find out that $\ket{\psi_0}$ cannot generate $P$ by local operations only, we immediately know that $\rho_0$ cannot generate $P$. In other words, a necessary condition for $\rho$ can generate $P$ is that \emph{any} specific purification $\ket{\psi}$ of $\rho_0$ in $\h_{A_1}\otimes \h_A\otimes \h_B \otimes \h_{B_1}$ can also generate $P$.

    \begin{example}
    Suppose the target correlation is $P=\begin{bmatrix}
        0.3 & 0\\
        0 & 0.7
        \end{bmatrix}$. Alice and Bob share a {two-qubit} mixed state $\rho=\frac{1}{2}\ket{0_A0_B}\bra{0_A0_B}+\frac{1}{2}\ket{1_A1_B}\bra{1_A1_B}$, and want to generate $P$ by local operations. It can be verified that the pure state $\ket{\psi} = \frac{1}{\sqrt{2}}\ket{0_{A_1}0_A0_B0_{B_1}}+\frac{1}{\sqrt{2}}\ket{1_{A_1}1_A1_B1_{B_1}}$ is a specific purification of $\rho$. According to Eq.(\ref{eq:first_condition}), we know that the square of the minimal Schmidt coefficient of the pure seed state must be less than or equal to 0.3, which rules out the possibility of generating $P$ from $\ket{\psi}$, hence the original mixed state $\rho$ {cannot} generate $P$.
    \end{example}

    \subsection{When the seed is also a classical correlation}

    In this subsection, we study the case that both the seed and target are classical correlations, denoted by $P=[P(x,y)]_{x,y}$ and $Q=[Q(x,y)]_{x,y}$, respectively. We will show that this case is still NP-hard.

    \subsubsection{Quantum has no advantage in reachability}

\
Suppose $P=[P(x,y)]_{x,y}$ and $Q=[Q(x,y)]_{x,y}$ are two different classical correlations with the same size. We now prove that quantum operation has no advantage in the task of transforming $P$ to $Q$ locally. {A similar result was also reported in \cite{chitambar2018conditional}.}
    \begin{theorem}\label{thm:noadvantage}
    If $P$ can generate $Q$ by local quantum operations, then the same can be achieved by local classical operations.
    \end{theorem}
    \begin{proof}
	Suppose $\{E_i\}$ and $\{F_j\}$ are the Kraus operators for the quantum operations performed by Alice and Bob respectively such that
	\begin{align*}
		\sum_{i,j}(E_i\otimes F_j)\rho(E_i^{\dagger}\otimes F_j^{\dagger})= & \sigma,
	\end{align*}
	where
	\begin{align*}
		\rho= & \sum_{x,y}P(x,y)\ket{x}\bra{x}\otimes\ket{y}\bra{y}, \\
		\sigma= & \sum_{x,y}Q(x,y)\ket{x}\bra{x}\otimes\ket{y}\bra{y}.
	\end{align*}
	We expand the above relation with three pairs of indices and obtain
	\begin{align*}
		\sigma= & \sum_{i,j,x,y,x',y',x'',y''}\!P(x,y)\qip{x'}{E_i|x}\qip{y'}{F_j|y}\qip{x}{E_i^{\dagger}|x''}\\
        &\qquad \qquad \qquad \qquad \quad\cdot \qip{y}{F_j^{\dagger}|y''}\ket{x'}\bra{x''}\otimes\ket{y'}\bra{y''} \\
		= & \sum_{i,j,x,y,x',y'}\!P(x,y)\qip{x'}{E_i|x}\qip{y'}{F_j|y}\qip{x}{E_i^{\dagger}|x'}\\
        &\qquad\qquad\qquad\quad\cdot\qip{y}{F_j^{\dagger}|y'}\ket{x'}\bra{x'}\otimes\ket{y'}\bra{y'} \\
		=&\sum_{i,j,x,y,x',y'}P(x,y)|\qip{x'}{E_i|x}|^2|\qip{y'}{F_j|y}|^2\ket{x'y'}\bra{x'y'}.
	\end{align*}
	This implies that
	\begin{align*}
		Q(x',y')= & \sum_{i,j,x,y}P(x,y)|\qip{x'}{E_i|x}|^2|\qip{y'}{F_j|y}|^2 \\
		=  \sum_{x,y}\!&P(x,y)\!\left(\!\sum_{i}|\qip{x'}{E_i|x}|^2\!\right)\!\left(\!\sum_{j}|\qip{y'}{F_j|y}^2|\!\right)\!.
	\end{align*}
	We define
	\begin{align*}
		\mathrm{P}(x'|x)= & \sum_{i}|\qip{x'}{E_i|x}|^2, \\
		\mathrm{P}(y'|y)= & \sum_{j}|\qip{y'}{F_j|y}|^2.
	\end{align*}
    Then it can be verified that both $\mathrm{P}(x'|x)$ and $\mathrm{P}(y'|y)$ are valid conditional probabilities. Thus if Alice and Bob take them as local classical operations, they can generate $Q$ based on the seed classical correlation $P$.

    \end{proof}

    \subsubsection{The NP-hardness of the case of classical seed}

    We now show that even if the seed state is restricted to a classical correlation, {determining} whether it can generate a target classical correlation is NP-{hard}.
	\begin{theorem}
		\label{thm:np-hard-classical}
		The problem of deciding whether a given correlation $P_1$ can generate another given correlation $P_2$ via local quantum operations is $\mathbf{NP}$-hard.
	\end{theorem}
	\begin{proof}
        {Similar to the proof of Theorem \ref{thm:np-hard-quantum}, we reduce the $\mathsf{SUBSET\mhyphen SUM}$ problem, which takes as input a set of positive integers $\{a_1, a_2, \ldots, a_r\}$, to this problem.} We consider the case that
		\begin{align*}
			P_1= & \mathrm{diag}(\lambda_1,\cdots,\lambda_r), \\
			P_2= & \frac{1}{2}I_{2\times 2}=\frac{1}{2}
			\begin{pmatrix}
				1 & 0 \\
				0 & 1 \\
			\end{pmatrix},
		\end{align*}
      	where {$\lambda_i = \frac{a_i}{\sum_ka_k}$, hence $\sum_{i=1}^r\lambda_i=1$.} According to Theorem \ref{thm:noadvantage}, we can {assume} that the local operations performed by Alice and Bob are classical. As a result, the correlation $P_2$ can be generated from $P_1$ if and only if there exists $2\times r$ matrices $A,B$ such that
		\begin{align*}
			A(i,j)\geq & 0, &\forall i\in[2], j\in[r], \\
			B(i,j)\geq & 0, &\forall i\in[2], j\in[r], \\
			A(1,j)+A(2,j)=& B(1,j)+B(2,j)=1,  \!\!\!\!\!\!\!&\forall j\in[r], \\
			P_2= & AP_1B^T.
		\end{align*}
		Suppose such $A$ and $B$ indeed exist.
		We now fix some $j\in[r]$ and suppose ${A(1,j)}>0$.
		By ${P_2(1,2)}=0$, we have ${B(2,j)}=0$, hence ${B(1,j)}=1$.
		Since ${P_2(2,1)}=0$, this implies that ${A(2,j)}=0$, hence ${A(1,j)}=1$.
		Due to the symmetry of the problem, a similar conclusion can be drawn for ${A(2,j)}$, ${B(1,j)}$, and ${B(2,j)}$.
		Therefore{,} there exists $S\subseteq [r]$ such that $(j\in S\implies {A(1,j)=B(1,j)}=1)$ and $(j\not\in S\implies {A(2,j)=B(2,j)=1})$.
		This gives ${P_2(1,1)}=\sum_{j\in S}\lambda_j=\frac{1}{2}$, hence {$\sum_{j\in S}a_j$ is} a solution to the corresponding $\mathsf{SUBSET\mhyphen SUM}$ instance. {Conversely, if the $\mathsf{SUBSET\mhyphen SUM}$ instance has a solution $\sum_{j\in S}a_j$, then let $A$ and $B$ be constructed samely as above, we obtain that $P_2=  AP_1B^T$.}

        Therefore, we successfully construct a reduction from the $\mathsf{SUBSET\mhyphen SUM}$ problem to the problem that generates $P_2$ from $P_1$, implying that the latter is NP-{hard}.

	\end{proof}

    \section{An algorithm that computes {diagonal} forms of PSD factorizations}

    According to Theorem \ref{thm:canonical}, to determine whether or not $\ket{\psi}$ can produce $P$ by local operations only, one needs to characterize all possible {diagonal} forms of PSD {factorizations} for $P$, and find out whether one of them is consistent with $\ket{\psi}$.
	In the current section, we propose an algorithm to compute the {diagonal} forms of PSD {factorizations} for an arbitrary correlation $P$ {with respect to a given diagonal matrix $\Lambda$}.

    For this, we now formulate the task of computing {diagonal} {forms} of PSD factorizations as an optimization problem, where the variables are two sets of matrices $\{C_x\}$ and $\{D_y\}$ belonging to the positive semidefinite cone:
    \begin{align}\label{eq:optimization_origin}
    \min _{\substack{C_x, D_y \geq 0 \\ x=1, \ldots, n \\ y=1, \ldots, m}} \sum_{x=1}^n \sum_{y=1}^m\left(P(x,y)-\tr(C_xD_y) \right)^2 \\
    \text{subject to  } \sum_{x=1}^n C_x=\sum_{y=1}^m D_y=\Lambda.\nonumber
    \end{align}
	This optimization problem is non-convex and NP-hard \cite{shitov2017complexity}.
	Since the constraints and objective function are polynomial functions of the variables,
	some established methods such as quantifier elimination \cite{tarski1998decision} and Lasserre's hierarchy \cite{lasserre2001global} could be employed to solve the optimization problem, albeit with exponential running time.
	To obtain a heuristic approach for this problem, we fix one of the two sets of matrices to simplify it to a convex optimization problem,
	similar to the algorithm in \cite{vandaele2018algorithms} for computing PSD factorization.
	We accordingly develop an algorithm that computes {diagonal} forms of PSD {factorizations} by alternately optimizing over $\{C_x\}$ and $\{D_y\}$.
	The pseudo-code is represented in {Algorithm \ref{alg_seesaw}}.
    \begin{algorithm}
        %\textsl{}\setstretch{1.8}
        \renewcommand{\algorithmicrequire}{\textbf{Input:}}
        \renewcommand{\algorithmicensure}{\textbf{Output:}}
        \caption{{Alternating optimization for the {diagonal} form of PSD factorizations}}
        \hspace*{0.05in}\textbf{INPUT:} $P \in \mathbb{R}_{+}^{n \times m}$, {initial} $\{D_y\}$,  {and a diagonal matrix $\Lambda$}.\\
        \hspace*{0.05in}\textbf{OUTPUT:}$\{C_x\}\text{ and } \{D_y\}$.
        \label{alg1}
        \begin{algorithmic}[1]
            \WHILE{stopping criterion not satisfied}
                \STATE $\{C_x\} \leftarrow$ sub-algorithm $\left(P,\{D_y\}\right)$
                \STATE $\{D_y\} \leftarrow$ sub-algorithm $\left(P^T,\{C_x\}\right)$
            \ENDWHILE
        \end{algorithmic}
        \label{alg_seesaw}
    \end{algorithm}

    Particularly, the initial input of each matrix $D_{y'} \in \{D_y\}$ is given by
    $$D_{y'}=\sum_{i=1}^r b^i b^{i^T}$$ where $b^i$'s are $n$-dimensional vector whose entries are initialized using the normal distribution $\mathcal{N}(0,1)$, and $r$ is the size of the {diagonal} {form of PSD factorizations}. Each sub-algorithm returns the solution of a convex optimization problem. Taking sub-algorithm $\left(P,\{D_y\}\right)$ for example, the corresponding convex optimization problem can be written as
    $$\min _{\substack{C_x \geq 0 \\ x=1, \ldots, n}} \sum_{x=1}^n \sum_{y=1}^m\left(P(x,y)-\tr(C_xD_y) \right)^2$$
    $$\text{subject to  } \sum_{x=1}^n C_x=\Lambda,$$
    where $\{D_y\}$ are supposed to be fixed. We use CVXPY to solve these subproblems {\cite{diamond2016cvxpy}}, and the pseudo-code is {given in Algorithm \ref{alg_sub}}.
    \begin{algorithm}
        %\textsl{}\setstretch{1.8}
        \renewcommand{\algorithmicrequire}{\textbf{Input:}}
        \renewcommand{\algorithmicensure}{\textbf{Output:}}
        \caption{Sub-algorithm}
        \hspace*{0.05in}\textbf {INPUT:} $P \in \mathbb{R}_{+}^{n \times m}$, {initial} $\{D_y\}$, {and a diagonal matrix $\Lambda$}.\\
        \hspace*{0.05in}\textbf {OUTPUT:} $\{C_x\}$.
        \label{alg2}
        \begin{algorithmic}[1]
            \STATE $\{C_x\} \!\leftarrow$\! arg$\min_{\substack{C_x \geq 0 \\ x=1, \ldots, n}} \!\displaystyle\sum_{x=1}^n \sum_{y=1}^m\left(P(x,y)\!-\!\tr(C_xD_y) \right)^2$\\
            \qquad \qquad \qquad \qquad \qquad $\text{subject to  } \displaystyle\sum_{x=1}^n C_x=\Lambda$
        \end{algorithmic}
        \label{alg_sub}
    \end{algorithm}

    We now see a nontrivial application of this algorithm. Suppose we need to determine whether a pure state $\ket{\psi}$ can generate a target classical correlation $P$. For this, we first test whether any one of the aforementioned necessary conditions we have found is violated by this case. If we find such a necessary condition, then we know that $\ket{\psi}$ cannot generate $P$; if no such necessary condition is found, for the time being we suppose it is possible to generate $P$ from $\ket{\psi}$ and then employ {Algorithm \ref{alg_seesaw}} to search for the possible protocols.

    For example, let $\ket{\psi}$ be a 2-qubit pure state with Schmidt coefficient $1/\sqrt{5}$ and $2/\sqrt{5}$, and $P=\frac{1}{3}\begin{bmatrix}
        1 & 1\\
        1 & 0
    \end{bmatrix}$. Note that all the necessary conditions we have listed are satisfied in this case. We set $\Lambda = \text{diag}(1/\sqrt{5},2/\sqrt{5})$ and $r = 2$, then run the algorithm. {The following diagonal form of PSD factorizations can be found in two seconds on a personal computer}:\\
    $$\begin{array}{ll}
        C_1=\left[\begin{array}{ll}
        0.26801401 & 0.22523125\\
        0.22523125 & 0.61132503
    \end{array}\right], \\
    \\
    C_2=\left[\begin{array}{ll}
        0.17919958 & -0.22523125\\
        -0.22523125 & 0.28310216
    \end{array}\right], \\
    \\
    D_1=\left[\begin{array}{ll}
        0.12072403 & -0.26145645\\
        -0.26145645 & 0.68499394
    \end{array}\right], \\
    \\
    D_2=\left[\begin{array}{ll}
        0.32648956 & 0.26145646\\
        0.26145646 & 0.20943325
    \end{array}\right],
    \end{array}$$\\
    where the objective function equals {$9.2\times10^{-10}$}. As a result, we believe that this seed state can be used to generate the target correlation under local operations.

    \subsection*{Acknowledgments}
    Z.W. thanks Jamie Sikora for the helpful discussions. {We would like to express our gratitude to the referees for their helpful comments and suggestions.} Z.C., L.L., X.L., and Z.W. are supported by the National Natural Science Foundation of China, Grant No. 62272259, 62332009, 61832015, the National Key R\&D Program of China, Grant No. 2018YFA0306703, 2021YFE0113100, and Beijing Natural Science Foundation, No. Z220002. P.Y. is supported  by National Natural Science Foundation of China (Grant No. 62332009, 61972191) and Innovation Program for Quantum Science and Technology (Grant No. 2021ZD0302901).

    \bibliographystyle{alpha}
  \bibliography{ref}

\newcommand{\etalchar}[1]{$^{#1}$}
\begin{thebibliography}{MLDS{\etalchar{+}}13}

\bibitem[Bei13]{Beigi:2013}
Salman Beigi.
\newblock A new quantum data processing inequality.
\newblock {\em Journal of Mathematical Physics}, 54(8):082202, 2013.

\bibitem[BPT12]{blekherman2012semidefinite}
Grigoriy Blekherman, Pablo~A Parrilo, and Rekha~R Thomas.
\newblock {\em Semidefinite optimization and convex algebraic geometry}.
\newblock SIAM, 2012.

\bibitem[CFH18]{chitambar2018conditional}
Eric Chitambar, Ben Fortescue, and Min-Hsiu Hsieh.
\newblock The conditional common information in classical and quantum secret
  key distillation.
\newblock {\em IEEE Transactions on Information Theory}, 64(11):7381--7394,
  2018.

\bibitem[CFL{\etalchar{+}}12]{chau2012entanglement}
HF~Chau, Chi-Hang~Fred Fung, Chi-Kwong Li, Edward Poon, and Nung-Sing Sze.
\newblock Entanglement transformation between sets of bipartite pure quantum
  states using local operations.
\newblock {\em Journal of mathematical physics}, 53(12):122201, 2012.

\bibitem[DB14]{Delgosha2014}
Payam Delgosha and Salman Beigi.
\newblock Impossibility of local state transformation via hypercontractivity.
\newblock {\em Communications in Mathematical Physics}, 332(1):449--476, Nov
  2014.

\bibitem[DB16]{diamond2016cvxpy}
Steven Diamond and Stephen Boyd.
\newblock {CVXPY}: A python-embedded modeling language for convex optimization.
\newblock {\em The Journal of Machine Learning Research}, 17(1):2909--2913,
  2016.

\bibitem[FGP{\etalchar{+}}15]{fawzi2015positive}
Hamza Fawzi, Jo{\~a}o Gouveia, Pablo~A Parrilo, Richard~Z Robinson, and Rekha~R
  Thomas.
\newblock Positive semidefinite rank.
\newblock {\em Mathematical Programming}, 153(1):133--177, 2015.

\bibitem[FMP{\etalchar{+}}12]{fiorini2012linear}
Samuel Fiorini, Serge Massar, Sebastian Pokutta, Hans~Raj Tiwary, and Ronald
  {d}e Wolf.
\newblock Linear vs. semidefinite extended formulations: exponential separation
  and strong lower bounds.
\newblock In {\em Proceedings of the forty-fourth annual ACM symposium on
  Theory of computing}, pages 95--106, 2012.

\bibitem[GK73]{Gacs:1973}
Peter G\'acs and J~K\"orner.
\newblock Common information is far less than mutual information.
\newblock {\em Problems of Control and Information Theory}, 2, Jan 1973.

\bibitem[GKS16]{7782969}
Badih Ghazi, Pritish Kamath, and Madhu Sudan.
\newblock Decidability of non-interactive simulation of joint distributions.
\newblock In {\em 2016 IEEE 57th Annual Symposium on Foundations of Computer
  Science}, pages 545--554, Los Alamitos, CA, USA, Oct 2016. IEEE Computer
  Society.

\bibitem[Har82]{hartmanis1982computers}
Juris Hartmanis.
\newblock Computers and intractability: a guide to the theory of
  {NP}-completeness ({M}ichael {R}. {G}arey and {D}avid {S}. {J}ohnson).
\newblock {\em {SIAM} {R}eview}, 24(1):90, 1982.

\bibitem[JSWZ13]{jain2013efficient}
Rahul Jain, Yaoyun Shi, Zhaohui Wei, and Shengyu Zhang.
\newblock Efficient protocols for generating bipartite classical distributions
  and quantum states.
\newblock {\em IEEE Transactions on Information Theory}, 59(8):5171--5178,
  2013.

\bibitem[KA16]{7452414}
S.~{Kamath} and V.~{Anantharam}.
\newblock On non-interactive simulation of joint distributions.
\newblock {\em IEEE Transactions on Information Theory}, 62(6):3419--3435, June
  2016.

\bibitem[Las01]{lasserre2001global}
Jean~B Lasserre.
\newblock Global optimization with polynomials and the problem of moments.
\newblock {\em SIAM Journal on optimization}, 11(3):796--817, 2001.

\bibitem[LWdW17]{lee2017some}
Troy Lee, Zhaohui Wei, and Ronald {d}e Wolf.
\newblock Some upper and lower bounds on {PSD}-rank.
\newblock {\em Mathematical programming}, 162(1-2):495--521, 2017.

\bibitem[MBG{\etalchar{+}}19]{mojahedian2019correlation}
Mohammad~Mahdi Mojahedian, Salman Beigi, Amin Gohari, Mohammad~Hossein Yassaee,
  and Mohammad~Reza Aref.
\newblock A correlation measure based on vector-valued ${L}_p$-norms.
\newblock {\em IEEE Transactions on Information Theory}, 65(12):7985--8004,
  2019.

\bibitem[MLDS{\etalchar{+}}13]{muller2013quantum}
Martin M{\"u}ller-Lennert, Fr{\'e}d{\'e}ric Dupuis, Oleg Szehr, Serge Fehr, and
  Marco Tomamichel.
\newblock On quantum r{\'e}nyi entropies: A new generalization and some
  properties.
\newblock {\em Journal of Mathematical Physics}, 54(12):122203, 2013.

\bibitem[NC11]{10.5555/1972505}
Michael~A. Nielsen and Isaac~L. Chuang.
\newblock {\em Quantum Computation and Quantum Information: 10th Anniversary
  Edition}.
\newblock Cambridge University Press, USA, 10th edition, 2011.

\bibitem[Pis98]{pisier1998non}
Gilles Pisier.
\newblock {\em Non-commutative vector valued $L_p$-spaces and completely
  $p$-summing maps}.
\newblock Soci{\'e}t{\'e} math{\'e}matique de France, 1998.

\bibitem[QY21]{qin2021nonlocal}
Minglong Qin and Penghui Yao.
\newblock Nonlocal games with noisy maximally entangled states are decidable.
\newblock {\em SIAM Journal on Computing}, 50(6):1800--1891, 2021.

\bibitem[QY23]{QYao:2022}
Minglong Qin and Penghui Yao.
\newblock Decidability of fully quantum nonlocal games with noisy maximally
  entangled states.
\newblock In {\em 50th International Colloquium on Automata, Languages, and
  Programming, to appear}, 2023.

\bibitem[Shi17]{shitov2017complexity}
Yaroslav Shitov.
\newblock The complexity of positive semidefinite matrix factorization.
\newblock {\em SIAM Journal on Optimization}, 27(3):1898--1909, 2017.

\bibitem[Tar98]{tarski1998decision}
Alfred Tarski.
\newblock {\em A decision method for elementary algebra and geometry}.
\newblock Springer, 1998.

\bibitem[Vav10]{vavasis2010complexity}
Stephen~A Vavasis.
\newblock On the complexity of nonnegative matrix factorization.
\newblock {\em SIAM Journal on Optimization}, 20(3):1364--1377, 2010.

\bibitem[VGG18]{vandaele2018algorithms}
Arnaud Vandaele, Fran{\c{c}}ois Glineur, and Nicolas Gillis.
\newblock Algorithms for positive semidefinite factorization.
\newblock {\em Computational Optimization and Applications}, 71(1):193--219,
  2018.

\bibitem[Wyn75]{Wyner:1975:CIT:2263311.2268812}
A.~Wyner.
\newblock The common information of two dependent random variables.
\newblock {\em IEEE Transactions on Information Theory}, 21(2):163--179, March
  1975.

\bibitem[Yan88]{yannakakis1988expressing}
Mihalis Yannakakis.
\newblock Expressing combinatorial optimization problems by linear programs.
\newblock In {\em Proceedings of the twentieth annual ACM symposium on Theory
  of computing}, pages 223--228, 1988.

\bibitem[Zha12]{zhang2012quantum}
Shengyu Zhang.
\newblock Quantum strategic game theory.
\newblock In {\em Proceedings of the 3rd Innovations in Theoretical Computer
  Science Conference}, pages 39--59, 2012.

\end{thebibliography}

\end{document}